\documentclass[10pt,english]{article}
\usepackage[T1]{fontenc}
\usepackage[latin9]{inputenc}
\usepackage{xcolor}
\usepackage{pdfcolmk}
\usepackage{amsthm}
\usepackage{amsmath}
\usepackage{amssymb}
\usepackage{graphicx}
\PassOptionsToPackage{normalem}{ulem}
\usepackage{ulem}

\makeatletter

\providecolor{lyxadded}{rgb}{0,0,1}
\providecolor{lyxdeleted}{rgb}{1,0,0}

\newcommand{\lyxaddress}[1]{
\par {\raggedright #1
\vspace{1.4em}
\noindent\par}
}
\theoremstyle{plain}
\newtheorem{thm}{\protect\theoremname}
  \theoremstyle{remark}
  \newtheorem{rem}[thm]{\protect\remarkname}
  \theoremstyle{definition}
  \newtheorem{example}[thm]{\protect\examplename}
  \theoremstyle{definition}
  \newtheorem{defn}[thm]{\protect\definitionname}
  \theoremstyle{plain}
  \newtheorem{lem}[thm]{\protect\lemmaname}
  \theoremstyle{remark}
  \newtheorem*{notation*}{\protect\notationname}
  \theoremstyle{plain}
  \newtheorem{conjecture}[thm]{\protect\conjecturename}

\usepackage[all]{xy}

\usepackage{multirow}

\usepackage{geometry}

\usepackage{cite}

\date{}

\def\frontmatter@abstractheading{}

\newcommand{\se}{\mathsf{s}_{even}}
\newcommand{\so}{\mathsf{s}_{odd}}

\newcommand{\cntx}{\mathsf{CNTX}}

\makeatother

\usepackage{babel}
  \providecommand{\conjecturename}{Conjecture}
  \providecommand{\definitionname}{Definition}
  \providecommand{\examplename}{Example}
  \providecommand{\lemmaname}{Lemma}
  \providecommand{\notationname}{Notation}
  \providecommand{\remarkname}{Remark}
\providecommand{\theoremname}{Theorem}

\begin{document}

\title{Contextuality in Three Types of Quantum-Mechanical Systems }

\author{Ehtibar N. Dzhafarov\textsuperscript{1}, Janne V. Kujala\textsuperscript{2},
and Jan-Åke Larsson\textsuperscript{3}}

\maketitle

\lyxaddress{\begin{center}
\textsuperscript{1}Purdue University, ehtibar@purdue.edu \\\textsuperscript{2}University
of Jyv\"askyl\"a, jvk@iki.fi\\\textsuperscript{3}Linköping University,
jan-ake.larsson@liu.se
\par\end{center}}
\begin{abstract}
\mbox{}

We present a formal theory of contextuality for a set of random variables
grouped into different subsets (contexts) corresponding to different,
mutually incompatible conditions. Within each context the random variables
are jointly distributed, but across different contexts they are stochastically
unrelated. The theory of contextuality is based on the analysis of
the extent to which some of these random variables can be viewed as
preserving their identity across different contexts when one considers
all possible joint distributions imposed on the entire set of the
random variables. We illustrate the theory on three systems of traditional
interest in quantum physics (and also in non-physical, e.g., behavioral
studies). These are systems of the{\normalsize{} }Klyachko-Can-Binicioglu-Shumovsky-type,
Einstein-Podolsky-Rosen-Bell-type, and Suppes-Zanotti-Leggett-Garg-type.
Listed in this order, each of them is formally a special case of the
previous one. For each of them we derive necessary and sufficient
conditions for contextuality while allowing for experimental errors
and contextual biases or signaling. Based on the same principles that
underly these derivations we also propose a measure for the degree
of contextuality and compute it for the three systems in question. 

\mbox{}


\markboth{Dzhafarov and Kujala}{Contextuality Generalized}
\end{abstract}

\section{Introduction}

A deductive mathematical theory is bound to begin with definitions
and/or axioms, and one is free not to accept them. We propose a certain
definition of contextuality which may or may not be judged ``good.''
Ultimately, its utility will be determined by whether it leads to
fruitful mathematical developments and interesting applications. Our
definition applies to situations where contextuality is traditionally
investigated in quantum physics: Klyachko-Can-Binicioglu-Shumovsky-type
systems of measurements \cite{Klyachko}, Einstein-Podolsky-Rosen-Bell-type
systems \cite{Bell1964,Bell1966,9CHSH,10CH,15Fine}, and Suppes-Zanotti-Leggett-Garg-type
systems \cite{11Leggett,SuppesZanotti1981}. We will refer to these
systems by abbreviations KCBS, EPRB, and SZLG, respectively. In the
absence of what we call ``inconsistency,'' our contextuality criteria
(necessary and sufficient conditions) coincide with the traditional
inequalities. But our criteria also apply to situations with measurement
errors, contextual biases, and interaction among jointly measured
physical properties (``signaling''). Moreover, the logic of constructing
our criteria of contextuality leads to a natural quantification of
the degree of contextuality in the three types of systems considered. 

This paper can be viewed as a companion one for Ref. \cite{KDL2014},
in which we prove a general criterion for contextuality in ``cyclic''
systems of which the systems just mentioned (KCBS, EPRB, and SZLG)
are special cases. However, we use here a different criterion for
contextuality in these three types of systems, whose advantage is
in that it is directly related to the notion of the degree of contextuality.
At the end of this paper we conjecture (see Remark 40) a generalization
of the criterion and the measure of the degree of contextuality to
all ``cyclic'' systems. 

The notion of probabilistic contextuality is usually understood to
be about ``sewing together'' random variables recorded under different
conditions. That is, it is viewed as answering the question: given
certain sets of jointly distributed random variables, can a joint
distribution be found for their union? The key aspect and difficulty
in answering this questions is that different sets of random variables
generally pairwise overlap, share some of their elements. In Ernst
Specker's \cite{1Specker} well-known example with three magic boxes
containing (or not containing) gems, which we present here in probabilistic
terms, we have three binary random variables, $A,B,C$, that can only
be recorded in pairs,
\begin{equation}
X=\left(A,B\right),\; Y=\left(B,C\right),\; Z=\left(A,C\right).\label{eq:specker primitive}
\end{equation}
That is, the joint distribution of $A$ and $B$ in $X$ is known,
and the same is true for the components of $Y$ and $Z$. We ask whether
there is a joint distribution of all three of them, $\left(A,B,C\right)$,
that agrees with the distributions of $X$, $Y$, and $Z$ as its
2-component marginals. In Specker's example the boxes are magically
rigged so that (assuming $A,B,C$ attain values +1/-1, denoting the
presence/absence of a gem in the respective box) 
\begin{equation}
\Pr\left[A=-B\right]=1,\;\Pr\left[-B=C\right]=1,\;\Pr\left[C=-A\right]=1,\label{eq:Specker_incompatible}
\end{equation}
which, obviously, precludes the existence of a jointly distributed
$\left(A,B,C\right)$. We may say then that the system of random variables
(\ref{eq:specker primitive}) exhibits contextuality.

On a deeper level of analysis, however, contextuality is better to
be presented as a problem of determining \emph{identities of the random
variables recorded under different conditions}. That is, it answers
the question: is this random variable (under this condition), say,
$A$ in $X,$ ``the same as'' that one (under another condition),
say, $A$ in $Y$, or is the former at least ``as close'' to the
latter as their distributions in the two pairs allow?

This deeper view is based on the principle we dubbed \emph{Contextuality-by-Default},
developed through a series of recent publications \cite{22DK2013PLOS1,21DK2014FFOP,DzhKujNeurodynamics,6DK2014PLOS1,7DK2014arxiv,Acacio_et_al,DK_LG_Bell_1,DK_Bell_LG2}.
According to this principle, any two random variables recorded under
different (i.e., mutually exclusive) conditions (treatments) are labeled
by these conditions and considered \emph{stochastically unrelated}
(defined on different sample spaces\emph{,} possessing no joint distribution).
Thus, in Specker's example with the magic boxes, we need to denote
the observed three pairs of random variables not as in (\ref{eq:specker primitive}),
but as 
\begin{equation}
X=\left(A_{X},B_{X}\right),\; Y=\left(B_{Y},C_{Y}\right),Z=\left(A_{Z},C_{Z}\right).\label{eq:specker relabeled}
\end{equation}
Of course, any other unique labeling making random variables in one
context distinct from random variables in another context would do
as well. The notion of stochastic unrelatedness within the framework
of the Kolmogorovian probability theory has been explored in the quantum-theoretic
literature, notably by A. Khrennikov \cite{Khr2005,Khr2008,Khr2009}. 

The use of this notion within the present conceptual framework is
based on the fact that stochastically unrelated random variables can
always be \emph{coupled} (imposed a joint distribution upon) \cite{5DK2013LNCS,6DK2014PLOS1,7DK2014arxiv,8NHMT,12DK2013ProcAMS}.
This can generally be done in multiple ways, and no couplings are
privileged a priori. For Specker's example, one constructs a random
6-tuple 
\begin{equation}
S=\left(A_{X},B_{X},B_{Y},C_{Y},A_{Z},C_{Z}\right)\label{eq:Specker_S}
\end{equation}
such that its 2-marginals $X,Y,Z$ in (\ref{eq:specker relabeled})
are consistent with the observed probabilities. In particular, they
should satisfy 
\begin{equation}
\Pr\left[A_{X}=-B_{X}\right]=1,\;\Pr\left[-B_{Y}=C_{Y}\right]=1,\;\Pr\left[C_{Z}=-A_{Z}\right]=1.\label{eq:Specker_compatible}
\end{equation}
Such a coupling $S$ can be constructed in an infinity of ways. To
match this representation with Specker's original meaning, we have
to impose additional constraints on the possible couplings. Namely,
we have to require that $S$ in \eqref{eq:Specker_S} be constructed
subject to the following ``identity hypothesis'': 
\begin{equation}
\Pr\left[A_{X}=A_{Z}\right]=\Pr\left[B_{X}=B_{Y}\right]=\Pr\left[C_{Y}=C_{Z}\right]=1.\label{eq:Specker_constraint}
\end{equation}
Such a coupling $S$, as we have already determined, does not exist,
and we can say that the system of the random variables (\ref{eq:specker relabeled})
exhibits contextuality with respect to the identity hypothesis (\ref{eq:Specker_constraint}). 

One might wonder whether this re-representation of the problem is
useful. Aren't the questions 
\begin{quotation}
``Let me see if I can `sew together' $\left(A,B\right)$, $\left(B,C\right)$,
and $\left(A,C\right)$ into a single $\left(A,B,C\right)$,''
\end{quotation}
and
\begin{quotation}
``Let me see if I can put together $\left(A_{X},B_{X}\right),$ $\left(B_{Y},C_{Y}\right)$,
and $\left(A_{Z},C_{Z}\right)$ into a single $S$ in \eqref{eq:Specker_S}
under the identity hypothesis \eqref{eq:Specker_constraint},'' 
\end{quotation}
aren't they one and the same question in two equivalent forms? Clearly,
they are. But there are two (closely related) advantages of the second
formulation:
\begin{enumerate}
\item It can be readily generalized by replacing the perfect identities
in \eqref{eq:Specker_constraint} with less stringent or altogether
different constraints; and
\item for any given constraint, if a coupling satisfying it does not exist,
this approach allows one to gauge how close one can get to satisfying
it, i.e., one has a principled way for constructing a measure for
the degree of contextuality the system exhibits.
\end{enumerate}
To illustrate these interrelated points on Specker's example, observe
that the identity hypothesis \eqref{eq:Specker_constraint} cannot
be satisfied if the system is ``inconsistently connected,'' i.e.,
if the marginal distribution of, say, $A_{X}$ is not the same as
that of $A_{Z}$. This may happen if the magic boxes somehow physically
communicate (e.g., the gem can be transposed from one of the boxes
being opened to another), and the probability of finding a gem in
the first box ($A=1$) is affected differently by the opening of the
second box (i.e., in context $X$) and of the third box (in context
$Z$). $A_{X}$ and $A_{Z}$ may have different distributions also
as a result of (perhaps magically induced) errors in correctly identifying
which of the two open boxes contains a gem: e.g., when boxes $i$
and $j$ are open ($i<j$), one may with some probability erroneously
see/record the gem contained in the $i$th box as being in the $j$th
box. We may speak of ``signaling'' between the boxes in the former
case, and of ``contextual measurement biases'' in the latter. In
either case, the requirement (\ref{eq:Specker_constraint}) cannot
be satisfied for $A_{X}$ and $A_{Z}$, and to determine this one
does not even have to look at the observed distributions of $\left(A_{X},B_{X}\right),$
$\left(B_{Y},C_{Y}\right)$, and $\left(A_{Z},C_{Z}\right)$. However,
in either of the two cases one can meaningfully ask: what is the maximum
possible value of $\Pr\left[A_{X}=A_{Z}\right]$ that is consistent
with the distributions of $A_{X}$ and $A_{Z}$ (and analogously for
$B_{X},B_{Y}$ and $C_{Y},C_{Z}$), and are these maximum possible
values consistent with the observed distributions of $\left(A_{X},B_{X}\right),$
$\left(B_{Y},C_{Y}\right)$, and $\left(A_{Z},C_{Z}\right)$? 

We will proceed now to formulate these ideas in a more rigorous way.

\section{Systems, Random Bunches, and Connections}

Let $X=\left(A,B,C,\ldots\right)$ be a (generalized) sequence%
\footnote{A sequence is an indexed set, and ``generalized'' means that the
indexing is not necessarily finite or countable. We try to keep the
notation simple, omitting technicalities. A sequence of random variables
that are jointly distributed is a random variable, if the latter term
is understood broadly, as anything with a well-defined probability
distribution, to include random vectors, random sets, random processes,
etc.%
} of jointly distributed random variables, called \emph{components}
of $X$. We will refer to $X$ as a \emph{(random) bunch}. Let $\mathfrak{S}$
be a set of random bunches 
\begin{equation}
X=\left(A_{X},A'_{X},A''_{X},\ldots\right),Y=\left(B_{Y},B'_{Y},B''_{Y},\ldots\right),Z=\left(C_{Z},C'_{Z},C''_{Z},\ldots\right),\ldots
\end{equation}
(of arbitrary cardinalities), with the property that they are \emph{pairwise
componentwise stochastically unrelated}. The term means that no component
of one random bunch is jointly distributed with any component of another. 
\begin{rem}
Intuitively, each random bunch corresponds to certain \emph{conditions}
under which (or \emph{contexts} in which) the components of the bunch
are jointly recorded; and the conditions corresponding to different
random bunches are \emph{mutually exclusive}.
\end{rem}
Any pair $\left\{ A,B\right\} $ such that $A$ and $B$ are components
of two \emph{distinct} random bunches in $\mathfrak{S}$ is called
a (simple) \emph{connection}. A set $\mathcal{\mathfrak{C}}$ of pairwise
disjoint connections is called a \emph{simple set of (simple) connections}. 
\begin{rem}
Intuitively, a connection indicates a pair of random variables $A$
and $B$ that represent ``the same'' physical property, because
of which, ideally, they should be ``one and the same'' random variable
in different contexts. However, the distributions of $A$ and $B$
may be different due to signaling (from other random variables in
their contexts) or due to contextual measurement biases. Note that
the elements $A$ and $B$ of a connection never co-occur, i.e., they
possess no joint distribution, and their ``identity'' therefore
can never be verified by observation.
\end{rem}
Together, $\left(\mathfrak{S},\mathfrak{C}\right)$ form a \emph{system}
(of measurements, or of random bunches). Without loss of generality,
we can assume that $\mathfrak{S}$ contains no ``non-participating''
bunches, i.e., each random bunch $X$ has at least one component $A$
that belongs to some connection $\left\{ A,B\right\} $ in $\mathcal{\mathfrak{C}}$.
In particular, if set $\mathfrak{C}$ is finite, then so is set $\mathfrak{S}$
(even if the number of components in some of the bunches in $\mathfrak{S}$
is not finite).
\begin{example}[KCBS-system]
 A KCBS-system \cite{Klyachko} consists of five pairs of binary
($\pm1$) random variables,
\begin{equation}
\mathfrak{S}=\left\{ \left(V_{1},W_{2}\right),\left(V_{2},W_{3}\right),\left(V_{3},W_{4}\right),\left(V_{4},W_{5}\right),\left(V_{5},W_{1}\right)\right\} .\label{eq:Klyachko S}
\end{equation}
Abstracting away from the physical meaning, the schematic picture
below shows five radius-vectors, each corresponding to a distinct
physical property represented by a binary random variable. They can
only be recorded in pairs (\ref{eq:Klyachko S}), and each of these
pairs corresponds to vertices connected by an edge of the pentagram.
In accordance with the Contextuality-by-Default principle, we label
each variable both by index $i\in\left\{ 1,\ldots,5\right\} $ indicating
the radius-vector (physical property) it corresponds to, and by the
context, defined by which of the two pairs it enters. We use notation
$V$$_{i}$ in one of these pairs and $W_{i}$ in another. For instance,
$i=2$ is used to label $V_{2}$ in the pair $\left(V_{2},W_{3}\right)$
and $W_{2}$ in the pair $\left(V_{1},W_{2}\right)$. With this notation,
the simple set of the connections of interest in this system is
\begin{equation}
\mathfrak{C}=\left\{ \left(V_{1},W_{1}\right),\left(V_{2},W_{2}\right),\left(V_{3},W_{3}\right),\left(V_{4},W_{4}\right),\left(V_{5},W_{5}\right)\right\} .
\end{equation}

\end{example}
\begin{center}
\includegraphics[scale=0.3]{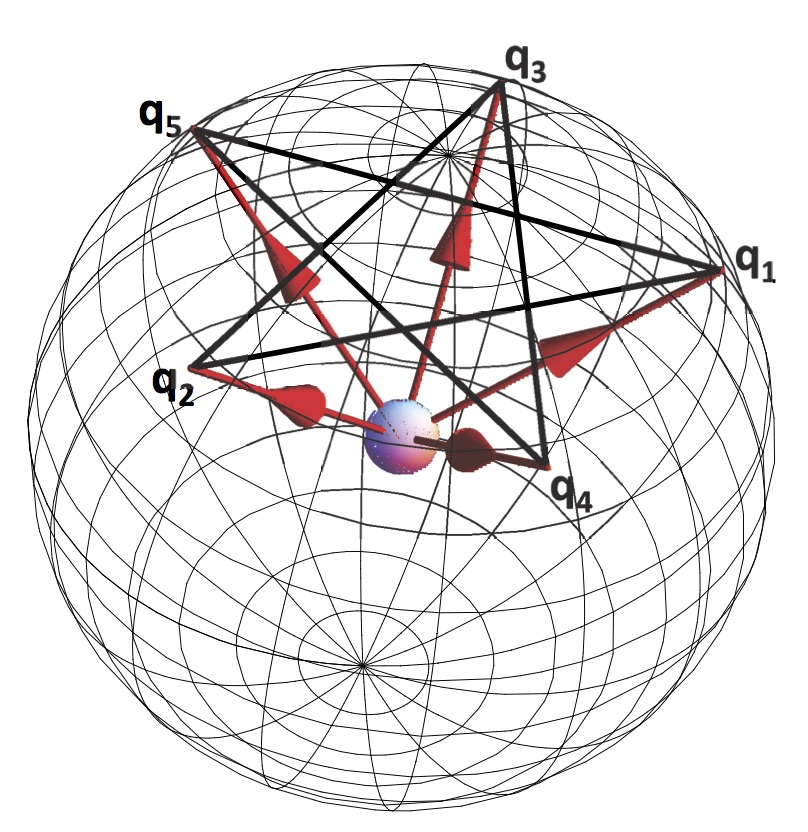}
\par\end{center}

In the ideal KCBS-system, each recorded pair, say, $\left(V_{1},W_{2}\right)$,
can attain values $\left(+1,-1\right)$, $\left(-1,+1\right)$, $\left(-1,-1\right)$,
but not $\left(+1,+1\right)$. In our analysis, however, we allow
for experimental errors, so that the ``pure'' KCBS-system is a special
case of a more general system in which $\Pr\left[V_{1}=+1,W_{2}=+1\right]$
may be non-zero. In the ideal KCBS-system the probabilities are computed
in accordance with the principles of quantum mechanics, so that the
distribution of $\left(V_{i},W_{j}\right)$ in (\ref{eq:Klyachko S})
is determined by the angle between the radius vectors $i$ and $j$,
and the distributions of $V_{i}$ is always the same as the distribution
of $W_{i}$ ($i=1,\ldots,5$). In our analysis, however, we allow
for ``signaling'' between the detectors and/or for ``contextual
measurement biases,'' so that, e.g., $V_{1}$ in $\left(V_{1},W_{2}\right)$
and $W_{1}$ in $\left(V_{5},W_{1}\right)$ may have different distributions.
\begin{rem}
There is no ``traditional'' contextual notation for the KCBS system,
but one can think of a variety of alternatives to our $V-W$ scheme,
e.g., denoting the $i$th measurement in the context of being conjoint
with the $j$th measurement by $R_{i}^{j}$, as we do in Ref. \cite{KDL2014}.\end{rem}
\begin{example}[EPRB-systems]
An EPRB-system \cite{Bell1964,Bell1966,9CHSH,10CH,15Fine} consists
of four pairs of binary ($\pm1$) random variables,
\begin{equation}
\mathfrak{S}=\left\{ \left(V_{1},W_{2}\right),\left(V_{2},W_{3}\right),\left(V_{3},W_{4}\right),\left(V_{4},W_{1}\right)\right\} .\label{eq:Alice-Bob S}
\end{equation}
Again we abstract away from the physical meaning, involving spins
of entangled particles. In the schematic picture below each direction
(1 or 3 in one particle and 2 or 4 in another) corresponds to a binary
random variable. They are recorded in pairs $\left\{ 1,3\right\} \times\left\{ 2,4\right\} $,
so each random variable participates in two contexts, and is denoted
either $V_{i}$ or $W_{i}$ ($i\in\left\{ 1,2,3,4\right\} $) accordingly.
The simple set of connections of interest is
\begin{equation}
\mathfrak{C}=\left\{ \left(V_{1},W_{1}\right),\left(V_{2},W_{2}\right),\left(V_{3},W_{3}\right),\left(V_{4},W_{4}\right)\right\} .
\end{equation}

\end{example}
\begin{center}
\includegraphics[scale=0.35]{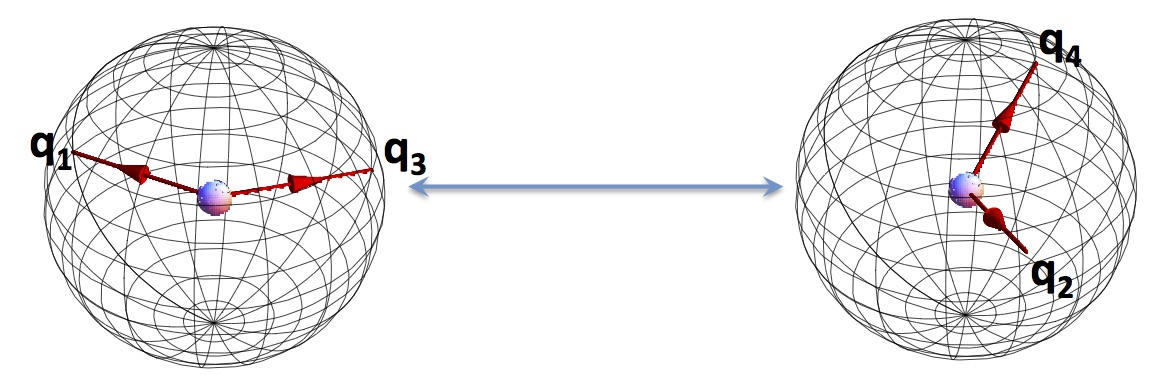}
\par\end{center}

Here, $\left(V_{i},W_{j}\right)$ may attain all four possible values
$\left(\pm1,\pm1\right)$. In the ideal system with space-like separation
between the recordings of $V_{i}$ and $W_{j}$, the distribution
of $V_{i}$ is always the same as that of $W_{i}$ ($i=1,\ldots,4$).
However, we allow for the possibility that the measurements are time-like
separated (so that direct signaling is possible), as well as for the
possibility that the results of the two measurements are recorded
by someone who may occasionally make errors and be contextually biased.
Thus, one may erroneously assign $+1$ to, say, $V_{1}=-1$ more often
than to $W$$_{1}=-1$. 
\begin{rem}
The contextual notation for the EPRB-systems adopted in our previous
papers \cite{22DK2013PLOS1,21DK2014FFOP,DzhKujNeurodynamics,6DK2014PLOS1,7DK2014arxiv,Acacio_et_al,DK_LG_Bell_1,DK_Bell_LG2}
is $\left(A_{ij},B_{ij}\right)$, $i,j\in\left\{ 1,2\right\} $, where
$A$ and $B$ refer to measurements on the first and second particles,
respectively. The first index refers to one of the two $A$-measurements
(1 or 2), the second index refers to one of the two $B$-measurements
(1 or 2). So the non-contextual (misleading) notation for $\left(A_{ij},B_{ij}\right)$
would be $\left(A_{i},B_{j}\right)$. In relation to our present notation,
$A_{11}$ corresponds to $V_{1}$, and $A_{12}$ (the same property
in another context) to $W_{1}$; $A_{21}$ corresponds to $W_{3}$,
and $A_{22}$ (the same property in another context) to $V_{3}$;
and analogously for $B_{ij}$. \end{rem}
\begin{example}[SZLG-system]
An SZLG-system \cite{11Leggett,SuppesZanotti1981} consists of three
pairs of binary ($\pm1$) random variables,
\begin{equation}
\mathfrak{S}=\left\{ \left(V_{1},W_{2}\right),\left(V_{2},W_{3}\right),\left(V_{3},W_{1}\right)\right\} .
\end{equation}
The three random variables are recorded in pairs, the logic of the
notation being otherwise the same as above. The simple set of connections
of interest in this system is
\begin{equation}
\mathfrak{C}=\left\{ \left(V_{1},W_{1}\right),\left(V_{2},W_{2}\right),\left(V_{3},W_{3}\right)\right\} .
\end{equation}
In the Leggett-Garg paradigm proper \cite{11Leggett}, the three measurements
are made at three moments of time, fixed with respect to some zero
point, as shown in the schematic picture below. This is, however,
only one possible physical meaning, and we can think of any three
identifiable measurements performed two at a time. 
\end{example}
\begin{center}
\includegraphics[scale=0.3]{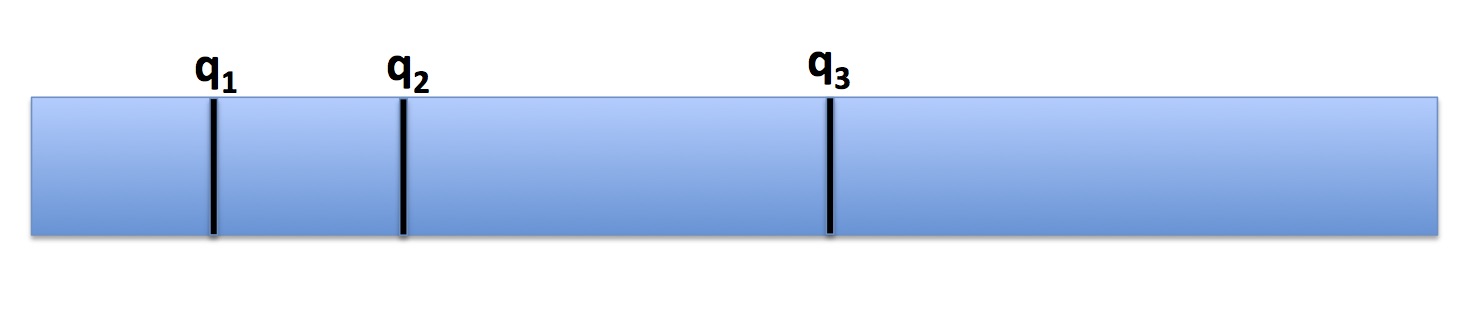}
\par\end{center}

Again, we allow for the possibility of signaling (which is predicted
by the laws of quantum mechanics in some cases, e.g., for pure initial
states, as shown in Ref. \cite{bacciagaluppi}), i.e., earlier measurements
may influence later ones. And, again, we allow for measurement errors
and contextual biases: knowing, e.g., that $W_{2}$ is preceded by
$V_{1}$ and is not followed by another measurement, and that $V_{2}$
is followed by $W_{2}$ and is not preceded by another measurement,
may lead one to record $V_{2}$ and $W_{2}$ differently even if they
are identically distributed ``in reality.''
\begin{rem}
The contextual notation for the LG-systems adopted in Ref. \cite{Acacio_et_al,DK_LG_Bell_1,DK_Bell_LG2}
is $\left(Q_{ij},Q_{ji}\right)$, $i,j\in\left\{ 1,2,3\right\} $
$\left(i<j\right)$, where the first index refers to the earlier of
the two measurements. Thus, $Q_{12}$ corresponds to $V_{1}$ in our
present notation, and $Q_{13}$ (the same property in another context)
to $W_{1}$; $Q_{23}$ corresponds to $V_{2}$, and $A_{21}$ (the
same property in another context) to $W_{2}$; and analogously for
$Q_{31}$ and $Q_{32}$ (resp., $V_{3}$ and $W_{3}$). In Ref. \cite{bacciagaluppi}
the notation used is $Q_{i}^{\left\{ i,j\right\} }$, where the superscript
indicates the context and the subscript the physical property.
\end{rem}

\section{Contextuality}

This section contains our main definitions: of a maximal connection,
of (in)consistent connectedness, and of contextuality.

\subsection{Couplings}
\begin{defn}
A \emph{coupling} of a set of random variables $X,Y,Z,\ldots$ is
a random bunch $\left(X{}^{*},Y{}^{*},Z{}^{*},\ldots\right)$ (with
jointly distributed components), such that 
\begin{equation}
X{}^{*}\sim X,Y{}^{*}\sim Y,Z{}^{*}\sim Z,\ldots,
\end{equation}
where $\sim$ stands for ``has the same distribution as.'' In particular,
a coupling $S$ for $\mathfrak{S}$ is a random bunch coupling all
elements (random bunches) of $\mathfrak{S}$. \end{defn}
\begin{example}
For two binary ($\pm1$) random variables $A,B$, any random bunch
$\left(A{}^{*},B{}^{*}\right)$ with the distribution
\begin{equation}
\begin{array}{l}
r_{11}=\Pr\left[A{}^{*}=+1,B{}^{*}=+1\right]\\
r_{10}=\Pr\left[A{}^{*}=+1,B{}^{*}=-1\right]\\
r_{01}=\Pr\left[A{}^{*}=-1,B{}^{*}=+1\right]\\
r_{00}=\Pr\left[A{}^{*}=-1,B{}^{*}=-1\right],
\end{array}
\end{equation}
such that
\begin{equation}
\begin{array}{c}
r_{11}+r_{10}=\Pr\left[A=+1\right],\\
r_{11}+r_{01}=\Pr\left[B=+1\right],
\end{array}
\end{equation}
is a coupling. \end{example}
\begin{rem}
It is a simple but fundamental theorem of Kolmogorov's probability
theory \cite{15Fine,8NHMT,DK2010,SuppesZanotti1981} that a coupling
$\left(X{}^{*},Y{}^{*},Z{}^{*},\ldots\right)$ of $X,Y,Z,\ldots$
exists if and only if there is a random variable $R$ and a sequence
of measurable functions $\left(f_{X},f_{Y},f_{Z},\ldots\right)$,
such that 
\begin{equation}
X\sim f_{X}\left(R\right),Y\sim f_{Y}\left(R\right),Z\sim f_{Z}\left(R\right),\ldots.
\end{equation}
In quantum mechanics, $R$ is referred to as a \emph{hidden variable}.
(In John Bell's pioneering work \cite{Bell1964}, he considers the
question of whether such a representation exists for four binary random
variables $A_{1},A_{2},B_{1},B_{2}$ with known distributions of $\left(A_{i},B_{j}\right)$,
$i,j\in\left\{ 1,2\right\} $. He imposes no constraints on $R$,
but it is easy to see that the existence of \emph{some} $R$ in his
problem is equivalent to the existence of an $R$ with just 16-values.)
\end{rem}

\subsection{Maximally Coupled Connections and Consistent Connectedness}
\begin{defn}
A coupling $\left(A{}^{*},B{}^{*}\right)$ of a connection $\left\{ A,B\right\} \in\mathfrak{C}$
is called \emph{maximal} if 
\begin{equation}
\Pr\left[A{}^{*}=B{}^{*}\right]\geq\Pr\left[A{}^{*}{}^{*}=B{}^{*}{}^{*}\right]
\end{equation}
for any coupling $\left(A{}^{**},B{}^{**}\right)$ of $\left\{ A,B\right\} $. 
\end{defn}

\begin{defn}
A system $\left(\mathfrak{S},\mathfrak{C}\right)$ is \emph{consistently
connected} (CC) if $A\sim B$ in any connection $\left\{ A,B\right\} \in\mathfrak{C}$.
Otherwise the system $\left(\mathfrak{S},\mathfrak{C}\right)$ is
\emph{inconsistently connected} (not CC). \end{defn}
\begin{rem}
In physics, the CC condition is sometimes referred to as ``no-signaling''
\cite{20Popescu,13Cereceda,14Masanes}, the term we are going to avoid
because then non-CC systems should be referred to as ``signaling.''
The latter term has strong connotations making its use in our technical
meaning objectionable to physicists. Inconsistent connectedness may
be due to signaling in the narrow physical meaning, but it may also
indicate measurement biases due to context (so that one measures $A$
differently when one also measures $B$ than when one also measures
$C$). We make no distinction between ``ideal'' random variables
and those measured ``incorrectly.''\end{rem}
\begin{lem}
In a CC system, a maximal coupling $\left(A{}^{*},B{}^{*}\right)$
for any connection $\left\{ A,B\right\} \in\mathfrak{C}$ exists,
and in this coupling $\Pr\left[A{}^{*}=B{}^{*}\right]=1$. If the
system is not CC, and $A\not\sim B$ in a connection $\left\{ A,B\right\} $,
then in a maximal coupling $\left(A^{*},B^{*}\right)$, if it exists,
$\Pr\left[A{}^{*}=B{}^{*}\right]<1$.
\end{lem}
A proof is obvious. We focus now on \emph{binary systems}, in which
all components of the random bunches are binary ($\pm1$) variables.
The three systems mentioned in the opening section, KCBS, LG, and
EPRB-type ones, are binary. (A generalization to components with finite
but arbitrary numbers of values is straightforward.)
\begin{lem}
\label{lem:Two binary}For a connection $\left\{ A,B\right\} $ with
binary ($\pm1$) $A,B$ and
\begin{equation}
p=\Pr\left[A=1\right]\geq\Pr\left[B=1\right]=q,
\end{equation}
a maximal coupling $\left(A{}^{*},B{}^{*}\right)$ exists, and its
distribution is
\begin{equation}
\begin{array}{l}
r_{11}=q\\
r_{10}=p-q\\
r_{01}=0\\
r_{00}=1-p,
\end{array}\label{eq:maximal coupling}
\end{equation}
where
\begin{equation}
r_{ab}=\Pr\left[A{}^{*}=2a-1,B{}^{*}=2b-1\right].
\end{equation}
\end{lem}
\begin{proof}
For given values $p\geq q$, the maximum possible value of $r_{11}$
is $\min\left(p,q\right)=q$, and the maximum possible value of $r_{00}$
is $\min\left(1-p,1-q\right)=1-p$; these values are attained in distribution
\eqref{eq:maximal coupling}, with $r_{01},r_{10}$ determined uniquely.
$\Pr\left[A{}^{*}=B{}^{*}\right]=r_{11}+r_{00}$ in this distribution
has the maximum possible value, $1-\left(p-q\right)$.\end{proof}
\begin{rem}
In a maximal coupling $\left(A{}^{*},B{}^{*}\right)$ of two binary
random variables $A,B$ the expectation
\begin{equation}
\left\langle A{}^{*}B{}^{*}\right\rangle =2\left(r_{11}+r_{00}\right)-1
\end{equation}
attains its maximum possible value (assuming $p\geq q$) 
\begin{equation}
\left\langle A{}^{*}B{}^{*}\right\rangle =1-2\left(p-q\right)=1-\left(\left\langle A\right\rangle -\left\langle B\right\rangle \right).
\end{equation}
. 
\end{rem}

\begin{rem}
\label{rem:In-many-cases}In some cases it is more convenient to speak
of the minimum value of $\Pr\left[A{}^{*}\not=B{}^{*}\right]=r_{10}+r_{01}$
rather than the maximum value of $\Pr\left[A{}^{*}=B{}^{*}\right]=r_{11}+r_{00}$.
This minimum can be presented as
\begin{equation}
\Pr\left[A{}^{*}\not=B{}^{*}\right]=p-q=\frac{1}{2}\left(\left\langle A\right\rangle -\left\langle B\right\rangle \right).
\end{equation}

\end{rem}

\begin{rem}
In a maximal coupling $\left(A{}^{*},B{}^{*}\right)$ of two binary
random variables $A,B$, 
\begin{equation}
\Pr\left[A{}^{*}=B{}^{*}\right]=1
\end{equation}
if and only if $A\sim B$, i.e.,
\begin{equation}
p=\Pr\left[A=1\right]=\Pr\left[B=1\right]=q.
\end{equation}

\end{rem}

\subsection{Contextuality}
\begin{defn}
\label{def:Contextuality}Let maximal couplings exist for all connections
in $\mathfrak{C}$. A system $\left(\mathfrak{S},\mathfrak{C}\right)$
has a \emph{(maximally) noncontextual description} if there exists
a coupling $S$ for $\mathfrak{S}$ in which all 2-marginals $\left(A{}^{*},B{}^{*}\right)$
that couple the connections in $\mathfrak{C}$ are maximal couplings.
If such a coupling $S$ does not exist, the system is \emph{contextual}.\end{defn}
\begin{rem}
We will omit the qualifier ``maximally'' when speaking of the existence
of a maximally noncontextual description.
\end{rem}

\begin{rem}
In particular, if the system is CC, it has a noncontextual description
if and only if there is a coupling $S$ for $\mathfrak{S}$ in which
$\Pr\left[A{}^{*}=B{}^{*}\right]=1$ for all connections $\left\{ A,B\right\} \in\mathfrak{C}$.
This is essentially the traditional use of the term ``(non)contextuality.''\end{rem}
\begin{example}
Let system $\left(\mathfrak{S},\mathfrak{C}\right)$ consist of random
bunches
\[
\left(A,C\right),\left(B,D\right),
\]
with all components binary ($\pm1$), and a single connection $\left\{ A,B\right\} $.
To determine if the system is contextual, we consider all possible
couplings for $A,B,C,D,$ i.e., all possible random bunches
\[
\left(A^{*},B^{*},C^{*},D^{*}\right)
\]
such that $\left(A^{*},C^{*}\right)\sim\left(A,C\right)$ and $\left(B^{*},D^{*}\right)\sim\left(B,D\right)$.
Denoting, for $a,b,c,d\in\left\{ 0,1\right\} $,
\begin{equation}
\begin{array}{c}
s_{ac}=\Pr\left[A=2a-1,C=2c-1\right],\\
t_{bd}=\Pr\left[B=2b-1,D=2d-1\right],\\
u_{abcd}=\Pr\left[A^{*}=2a-1,B^{*}=2b-1,C^{*}=2c-1,D^{*}=2d-1\right],
\end{array}
\end{equation}
the distributional equations $\left(A^{*},C^{*}\right)\sim\left(A,C\right)$
and $\left(B^{*},D^{*}\right)\sim\left(B,D\right)$ translate into
the following 8 equations for 16 probabilities $u_{abcd}$:
\begin{equation}
\sum_{b=0}^{1}\sum_{d=0}^{1}u_{abcd}=s_{ac},\quad a,c\in\left\{ 0,1\right\} ,\label{eq:basicexample1}
\end{equation}
\begin{equation}
\sum_{a=0}^{1}\sum_{c=0}^{1}u_{abcd}=t_{bd},\quad b,d\in\left\{ 0,1\right\} .\label{eq:basicexample2}
\end{equation}
The requirement that the coupling for $\left\{ A,B\right\} $ be maximal
translates into the additional four equations 
\begin{equation}
\sum_{c=0}^{1}\sum_{d=0}^{1}u_{abcd}=r_{ab}=\Pr\left[A^{*}=2a-1,B^{*}=2b-1\right],\label{eq:basicexample3}
\end{equation}
where, in view of Lemma \eqref{lem:Two binary}, $r{}_{ab}$ is given
by \eqref{eq:maximal coupling}, with the same meaning of $p,q$ and
the same convention $p=\Pr\left[A=1\right]\geq\Pr\left[B=1\right]=q.$
The problem of contextuality therefore reduces to one of determining
whether the system of 12 equations \eqref{eq:basicexample1}-\eqref{eq:basicexample2}-\eqref{eq:basicexample3}
for the 16 unknown $u_{abcd}\geq0$ has a solution. The answer in
this case can be shown to be affirmative, so the system considered
has a noncontextual description.
\end{example}
Definition \ref{def:Contextuality} is sufficient for all subsequent
considerations in this paper, but we note that it can be extended
to situations when maximal couplings for connections do not necessarily
exist. Let us associate with each connection for $A,B$ a supremal
number 
\[
p_{AB}=\sup_{\textnormal{all couplings }\left(A^{*},B^{*}\right)}\Pr\left[A{}^{*}=B{}^{*}\right].
\]

\begin{defn}[extended]
 A system $\left(\mathfrak{S},\mathfrak{C}\right)$ has a \emph{(maximally)
noncontextual} description (\emph{is contextual}) if there exists
(resp., does not exist) a sequence of couplings $S_{1},S_{2},\ldots$
for $\mathfrak{S}$ in which $\Pr\left[A{}^{*}=B{}^{*}\right]$ for
all connections $\left\{ A,B\right\} $ in $\mathfrak{C}$ uniformly
converge to the corresponding supremal numbers $p_{AB}$.
\end{defn}

\subsection{Measure of Contextuality for Binary Systems with Finite Simple Sets
of Connections}

We will assume that in the binary systems we are dealing with the
simple set of connections $\mathfrak{C}$ is finite:
\begin{equation}
\mathfrak{C}=\left\{ \left\{ A_{i},B_{i}\right\} :i\in\left\{ 1,\ldots,n\right\} \right\} .
\end{equation}
 
\begin{lem}
Given a finite simple set of connections $\left\{ \left\{ A_{i},B_{i}\right\} :i\in\left\{ 1,\ldots,n\right\} \right\} $
in a binary system, the respective couplings in the set $\left\{ \left(A_{i}^{*},B_{i}^{*}\right):i\in\left\{ 1,\ldots,n\right\} \right\} $
are all maximal if and only if 
\begin{equation}
\sum_{i=1}^{n}\Pr\left[A{}_{i}^{*}\not=B_{i}{}^{*}\right]=\frac{1}{2}\sum_{i=1}^{n}\left|\left\langle A_{i}^{*}\right\rangle -\left\langle B_{i}^{*}\right\rangle \right|.
\end{equation}
\end{lem}
\begin{proof}
Immediately follows from Lemma \eqref{lem:Two binary} and Remark
\eqref{rem:In-many-cases}.\end{proof}
\begin{notation*}
We denote 
\begin{equation}
\Delta_{0}\left(\mathfrak{C}\right)=\frac{1}{2}\sum_{i=1}^{n}\left|\left\langle A_{i}^{*}\right\rangle -\left\langle B_{i}^{*}\right\rangle \right|,
\end{equation}
and this quantity is to play a central role in the subsequent computations.\end{notation*}
\begin{defn}
Let $\Delta_{\min}\left(\mathfrak{S},\mathfrak{C}\right)$ for a system
with $\mathfrak{C}=\left\{ \left\{ A_{i},B_{i}\right\} :i\in\left\{ 1,\ldots,n\right\} \right\} $
be the infimum for 
\[
\sum_{i=1}^{n}\Pr\left[A{}_{i}^{*}\not=B_{i}{}^{*}\right]
\]
across all possible couplings $S$ for $\mathfrak{S}$. 
\end{defn}
This is another quantity to play a central role in subsequent computations.
\begin{thm}
For a binary system with a finite simple set of connections, the value
$\Delta_{\min}\left(\mathfrak{S},\mathfrak{C}\right)$ is achieved
in some coupling $S$, and
\begin{equation}
\Delta_{\min}\left(\mathfrak{S},\mathfrak{C}\right)\geq\Delta_{0}\left(\mathfrak{C}\right).\label{eq:delta-delta1}
\end{equation}
The system has a noncontextual description if and only if
\begin{equation}
\Delta_{\min}\left(\mathfrak{S},\mathfrak{C}\right)=\Delta_{0}\left(\mathfrak{C}\right).\label{eq:delta-delta2}
\end{equation}
\end{thm}
\begin{proof}
That $\Delta_{\min}\left(\mathfrak{S},\mathfrak{C}\right)$ is an
achievable minimum follows from the fact that any coupling $S$ is
described by a system of linear inequalities relating to each other
\[
\Pr\left[A^{*}=a,B^{*}=b,C^{*}=c,\ldots\right]
\]
for all possible values $\left(a,b,c,\ldots\right)$ of all the random
variables involved (the union of all random bunches in $\mathfrak{S}$).
$\Delta_{\min}\left(\mathfrak{S},\mathfrak{C}\right)$ being a linear
combinations of these probabilities, its infimum has to be a minimum.
\eqref{eq:delta-delta1} and \eqref{eq:delta-delta2} are obvious.
\end{proof}
This theorem allows one to construct a convenient definition for the
degree of contextuality in binary systems with finite number of connections.
\begin{defn}
In a binary system $\left(\mathfrak{S},\mathfrak{C}\right)$ with
a finite simple set of connections the \emph{degree of contextuality}
is 
\begin{equation}
\cntx\left(\mathfrak{S},\mathfrak{C}\right)=\Delta_{\min}\left(\mathfrak{S},\mathfrak{C}\right)-\Delta_{0}\left(\mathfrak{C}\right)\geq0.
\end{equation}

\end{defn}
In the subsequent sections of this paper we show how this definition
of contextuality applies to KCBS, LG, and EPRB-systems.
\begin{rem}
Clearly, a measure of contextuality could also be constructed as $\left(1+\Delta_{\min}\left(\mathfrak{S},\mathfrak{C}\right)\right)/\left(1+\Delta_{0}\left(\mathfrak{C}\right)\right)-1$,
$\left(\Delta_{\min}\left(\mathfrak{S},\mathfrak{C}\right)-\Delta_{0}\left(\mathfrak{C}\right)\right)/\left(\Delta_{\min}\left(\mathfrak{S},\mathfrak{C}\right)+\Delta_{0}\left(\mathfrak{C}\right)\right)$,
and in a variety of other ways. The only logically necessary aspect
of the definition is that $\cntx\left(\mathfrak{S},\mathfrak{C}\right)$
is zero when $\Delta_{\min}\left(\mathfrak{S},\mathfrak{C}\right)=\Delta_{0}\left(\mathfrak{C}\right)$
and positive otherwise. The simple difference is chosen because it
has been shown to have certain desirable properties \cite{Acacio_et_al},
but this choice is not critical for the present paper.
\end{rem}

\subsection{Conventions}

\subsubsection{Abuse of language}

To simplify notation we adopt the following convention: in a coupling
$\left(X^{*},Y^{*}\right)$ of two random variables $X,Y$ we drop
the asterisks and write simply $\left(X,Y\right)$. The abuse of language
thus introduced is common, if not universally accepted in quantum
physics, and we too conveniently resorted to it when discussing Specker's
magic boxes in our introductory section.

\subsubsection{Functions $\se$ and $\so$}

We will make use of the following notation. For any finite sequence
of real numbers $\left(a_{i}:i\in\left\{ 1,\ldots,n\right\} \right)$
we denote by
\begin{equation}
\se\left(a_{i}:i\in\left\{ 1,\ldots,n\right\} \right)=\se\left(a_{1},\ldots,a_{n}\right)=\max_{\textnormal{even number of }-\textnormal{'s}}\sum_{i=1}^{n}\left(\pm a_{i}\right),
\end{equation}
where each $\pm$ should be replaced with $+$ or $-$, and the maximum
is taken over all combinations of the signs containing an even number
of $-$'s.

Analogously,
\begin{equation}
\so\left(a_{i}:i\in\left\{ 1,\ldots,n\right\} \right)=\so\left(a_{1},\ldots,a_{n}\right)=\max_{\textnormal{odd number of }-\textnormal{'s}}\sum_{i=1}^{n}\left(\pm a_{i}\right).
\end{equation}

\begin{lem}
For any finite sequence of real numbers $\left(a_{i}:i\in\left\{ 1,\ldots,n\right\} \right),$
\begin{equation}
\se\left(a_{1},\ldots,a_{n}\right)=\sum\left|a_{i}\right|-2[a{}_{1}\ldots a_{k}<0]\min\left(\left|a_{1}\right|,\ldots,\left|a_{n}\right|\right),
\end{equation}
\begin{equation}
\so\left(a_{1},\ldots,a_{n}\right)=\sum\left|a_{i}\right|-2[a{}_{1}\ldots a_{k}>0]\min\left(\left|a_{1}\right|,\ldots,\left|a_{n}\right|\right).
\end{equation}

\end{lem}

\section{KCBS-systems}

\subsection{Main Theorem}
\begin{thm}[contextuality measure and criterion for KCBS-systems]
\label{thm:main B} In a KCBS-system $\left(\mathfrak{S},\mathfrak{C}\right)$,
with
\begin{equation}
\mathfrak{S}=\left\{ \left(V_{i},W_{i\oplus_{5}1}\right):i\in\left\{ 1,\ldots,5\right\} \right\} ,\mathfrak{C}=\left\{ \left(V_{i},W_{i}\right):i\in\left\{ 1,\ldots,5\right\} \right\} ,
\end{equation}
\textup{\emph{where $\oplus_{5}$ stands for circular addition of
$1$}} on $\left\{ 1,2,3,4,5\right\} $,
\begin{equation}
\Delta_{0}\left(\mathfrak{C}\right)=\frac{1}{2}\sum_{i=1}^{5}\left|\left\langle V_{i}\right\rangle -\left\langle W_{i}\right\rangle \right|,
\end{equation}
 
\begin{equation}
\Delta_{\min}\left(\mathfrak{S},\mathfrak{C}\right)=\frac{1}{2}\max\left(2\Delta_{0}\left(\mathfrak{C}\right),\;\so\left(\left\langle V_{i}W_{i\oplus_{5}1}\right\rangle :i\in\left\{ 1,\ldots,5\right\} \right)-3\right).
\end{equation}
Consequently, the degree of contextuality in the KCBS-system is
\begin{equation}
\cntx\left(\mathfrak{S},\mathfrak{C}\right)=\frac{1}{2}\max\left(0,\;\so\left(\left\langle V_{i}W_{i\oplus_{5}1}\right\rangle :i\in\left\{ 1,\ldots,5\right\} \right)-3-\sum_{i=1}^{5}\left|\left\langle V_{i}\right\rangle -\left\langle W_{i}\right\rangle \right|\right),
\end{equation}
and the system has a noncontextual description if and only if\emph{
}\textup{\emph{
\begin{equation}
\so\left(\left\langle V_{i}W_{i\oplus_{5}1}\right\rangle :i\in\left\{ 1,\ldots,5\right\} \right)\leq3+\sum_{i=1}^{5}\left|\left\langle V_{i}\right\rangle -\left\langle W_{i}\right\rangle \right|.
\end{equation}
}}\end{thm}
\begin{proof}
The computer-assisted proof is based on Lemma \ref{lem:1-2} below,
and its details, omitted here, are analogous to those in the proofs
of Theorems 3-6 in Appendix of Ref. \cite{DK_Bell_LG2}.\end{proof}
\begin{lem}
\label{lem:1-2}The necessary and sufficient condition for the connection
couplings $\left\{ \left(V_{i},W_{i}\right):i\in\left\{ 1,\ldots,5\right\} \right\} $
to be compatible with the observed pairs $\left\{ \left(V_{i},W_{i\oplus_{5}1}\right):i\in\left\{ 1,\ldots,5\right\} \right\} $
is 
\begin{equation}
\so\left(\left\langle V_{i}W_{i\oplus_{5}1}\right\rangle ,\left\langle V_{i}W_{i}\right\rangle :i\in\left\{ 1,\ldots,5\right\} \right)\leq8,
\end{equation}
which can be equivalently written as
\[
\begin{array}{l}
\so\left(\left\langle V_{i}W_{i\oplus_{4}1}\right\rangle :i\in\left\{ 1,\ldots,5\right\} \right)+\se\left(\left\langle V_{i}W_{i}\right\rangle :i\in\left\{ 1,\ldots,5\right\} \right)\le8,\\
\\
\se\left(\left\langle V_{i}W_{i\oplus_{4}1}\right\rangle :i\in\left\{ 1,\ldots,5\right\} \right)+\so\left(\left\langle V_{i}W_{i}\right\rangle :i\in\left\{ 1,\ldots,5\right\} \right)\le8.
\end{array}
\]
\end{lem}
\begin{rem}
The compatibility in the formulation of the lemma means the existence
of a coupling $S$ for $\mathfrak{S}$ with given marginals $\left\{ \left(V_{i},W_{i}\right):i\in\left\{ 1,\ldots,5\right\} \right\} $
and given (coupled) connections $\left\{ \left(V_{i},W_{i\oplus_{5}1}\right):i\in\left\{ 1,\ldots,5\right\} \right\} $.
\end{rem}

\subsection{Special cases}

In a CC KCBS-system, $\Delta_{0}\left(\mathfrak{C}\right)=0$, the
criterion for noncontextuality acquires the form
\begin{equation}
\so\left(\left\langle V_{i}W_{i\oplus_{5}1}\right\rangle :i\in\left\{ 1,\ldots,5\right\} \right)\leq3.\label{eq:Klyachko w/o signaling}
\end{equation}
If, in addition, the KCBS-exclusion is satisfied,%
\footnote{This means that the directions in the 3D real Hilbert space are chosen
strictly in accordance with \cite{Klyachko}, with no experimental
errors, signaling, or contextual biases involved. In this case, the
values $\left(+1,+1\right)$ for paired measurements $\left(V_{i},W_{i\oplus_{5}1}\right)$
are excluded by quantum theory.%
} i.e., in every $\left\langle V_{i}W_{i\oplus_{5}1}\right\rangle $,
\begin{equation}
\Pr\left[V_{i}=1,W_{i\oplus_{5}1}=1\right]=0,
\end{equation}
then we have
\begin{equation}
\left\langle V_{i}W_{i\oplus_{5}1}\right\rangle =1-2\left(p_{i}+p_{i\oplus_{5}1}\right),
\end{equation}
where 
\begin{equation}
p_{i}=\Pr\left[V_{i}=1\right]=\Pr\left[W_{i}=1\right],\quad i\in\left\{ 1,\ldots,5\right\} .
\end{equation}
It follows that
\begin{equation}
\sum_{i=1}^{5}p_{i}\leq2.\label{eq:Klyachko inequality}
\end{equation}
This is the KCBS inequality, that has been derived in Ref. \cite{Klyachko}
as a necessary condition for noncontextuality. As it turns out (we
omit the simple proof), this condition is also necessary.
\begin{thm}
In a CC KCBS-system with KCBS exclusion, \eqref{eq:Klyachko w/o signaling}
is equivalent to \eqref{eq:Klyachko inequality}.
\end{thm}

\section{EPRB-systems}

\subsection{Main Theorem}
\begin{thm}[contextuality measure and criterion for EPRB-systems]
\label{thm:main B-1} In an EPRB-system $\left(\mathfrak{S},\mathfrak{C}\right)$,
with

\begin{equation}
\mathfrak{S}=\left\{ \left(V_{i},W_{i\oplus_{4}1}\right):i\in\left\{ 1,\ldots,4\right\} \right\} ,\mathfrak{C}=\left\{ \left(V_{i},W_{i}\right):i\in\left\{ 1,\ldots,4\right\} \right\} ,
\end{equation}
\textup{\emph{where $\oplus_{4}$ stands for circular addition of
$1$ on }}$\left\{ 1,2,3,4\right\} $,
\begin{equation}
\Delta_{0}\left(\mathfrak{C}\right)=\frac{1}{2}\sum_{i=1}^{4}\left|\left\langle V_{i}\right\rangle -\left\langle W_{i}\right\rangle \right|,
\end{equation}
and 
\begin{equation}
\Delta_{\min}\left(\mathfrak{S},\mathfrak{C}\right)=\frac{1}{2}\max\left(2\Delta_{0}\left(\mathfrak{C}\right),\;\so\left(\left\langle V_{i}W_{i\oplus_{4}1}\right\rangle :i\in\left\{ 1,\ldots,4\right\} \right)-2\right).
\end{equation}
Consequently, the degree of contextuality in the EPRB-system is
\begin{equation}
\cntx\left(\mathfrak{S},\mathfrak{C}\right)=\frac{1}{2}\max\left(0,\so\left(\left\langle V_{i}W_{i\oplus_{4}1}\right\rangle :i\in\left\{ 1,\ldots,4\right\} \right)-2-\sum_{i=1}^{4}\left|\left\langle V_{i}\right\rangle -\left\langle W_{i}\right\rangle \right|\right),
\end{equation}
and the system has a noncontextual description if and only if\emph{
}\textup{\emph{
\begin{equation}
\so\left(\left\langle V_{i}W_{i\oplus_{4}1}\right\rangle :i\in\left\{ 1,\ldots,4\right\} \right)\leq2+\sum_{i=1}^{4}\left|\left\langle V_{i}\right\rangle -\left\langle W_{i}\right\rangle \right|.
\end{equation}
}}\end{thm}
\begin{proof}
The computer-assisted proof is based on Lemma \ref{lem:1-2-1} below,
and its details, omitted here, can be found in Ref. \cite{DK_Bell_LG2},
Appendix and Theorems 3 and 5.\end{proof}
\begin{lem}
\label{lem:1-2-1}The necessary and sufficient condition for the connections
$\left\{ \left(V_{i},W_{i}\right):i\in\left\{ 1,\ldots,4\right\} \right\} $
to be compatible with the observed pairs $\left\{ \left(V_{1},W_{2}\right),\left(V_{2},W_{3}\right),\left(V_{3},W_{4}\right),\left(V_{4},W_{1}\right)\right\} $
is
\begin{equation}
\so\left(\left\langle V_{i}W_{i\oplus_{5}1}\right\rangle ,\left\langle V_{i}W_{i}\right\rangle :i\in\left\{ 1,\ldots,4\right\} \right)\leq6,
\end{equation}
which can be equivalently written as
\begin{equation}
\begin{array}{l}
\so\left(\left\langle V_{i}W_{i\oplus_{4}1}\right\rangle :i\in\left\{ 1,\ldots,4\right\} \right)+\se\left(\left\langle V_{i}W_{i}\right\rangle :i\in\left\{ 1,\ldots,4\right\} \right)\le6,\\
\\
\se\left(\left\langle V_{i}W_{i\oplus_{4}1}\right\rangle :i\in\left\{ 1,\ldots,4\right\} \right)+\so\left(\left\langle V_{i}W_{i}\right\rangle :i\in\left\{ 1,\ldots,4\right\} \right)\le6.
\end{array}\label{eq:compatibility}
\end{equation}
 
\end{lem}

\subsection{Special case}

In a CC EPRB-system, $\Delta_{0}\left(\mathfrak{C}\right)=0$, the
criterion for noncontextuality acquires the form
\begin{equation}
\so\left(\left\langle V_{i}W_{i\oplus_{4}1}\right\rangle :i\in\left\{ 1,\ldots,4\right\} \right)\leq2,
\end{equation}
which is the standard CHSH inequalities \cite{10CH,15Fine,9CHSH},
presentable in a more familiar way as\emph{
\begin{equation}
\begin{array}{c}
-2\leq\left\langle V_{1}W_{2}\right\rangle +\left\langle V_{2}W_{3}\right\rangle +\left\langle V_{3}W_{4}\right\rangle -\left\langle V_{4}W_{1}\right\rangle \leq2,\\
-2\leq\left\langle V_{1}W_{2}\right\rangle +\left\langle V_{2}W_{3}\right\rangle -\left\langle V_{3}W_{4}\right\rangle +\left\langle V_{4}W_{1}\right\rangle \leq2,\\
-2\leq\left\langle V_{1}W_{2}\right\rangle -\left\langle V_{2}W_{3}\right\rangle +\left\langle V_{3}W_{4}\right\rangle +\left\langle V_{4}W_{1}\right\rangle \leq2,\\
-2\leq-\left\langle V_{1}W_{2}\right\rangle +\left\langle V_{2}W_{3}\right\rangle +\left\langle V_{3}W_{4}\right\rangle +\left\langle V_{4}W_{1}\right\rangle \leq2.
\end{array}
\end{equation}
}

\section{SZLG-systems}

\subsection{Main Theorem}
\begin{thm}[contextuality measure and criterion for SZLG-systems]
\label{thm:main B-1-1} In an SZLG-system $\left(\mathfrak{S},\mathfrak{C}\right)$,
with
\begin{equation}
\mathfrak{S}=\left\{ \left(V_{i},W_{i\oplus_{3}1}\right):i\in\left\{ 1,2,3\right\} \right\} ,\mathfrak{C}=\left\{ \left(V_{i},W_{i}\right):i\in\left\{ 1,2,3\right\} \right\} ,
\end{equation}
\textup{\emph{where $\oplus_{3}$ stands for circular addition of
$1$ on }}$\left\{ 1,2,3\right\} $,
\begin{equation}
\Delta_{0}\left(\mathfrak{C}\right)=\frac{1}{2}\sum_{i=1}^{3}\left|\left\langle V_{i}\right\rangle -\left\langle W_{i}\right\rangle \right|,
\end{equation}
 
\begin{equation}
\Delta_{\min}\left(\mathfrak{S},\mathfrak{C}\right)=\frac{1}{2}\max\left(2\Delta_{0}\left(\mathfrak{C}\right),\so\left(\left\langle V_{i}W_{i\oplus_{3}1}\right\rangle :i\in\left\{ 1,2,3\right\} \right)-1\right).
\end{equation}
Consequently, the degree of contextuality in the SZLG-system is
\begin{equation}
\cntx\left(\mathfrak{S},\mathfrak{C}\right)=\frac{1}{2}\max\left(0,\so\left(\left\langle V_{i}W_{i\oplus_{3}1}\right\rangle :i\in\left\{ 1,2,3\right\} \right)-1-\sum_{i=1}^{3}\left|\left\langle V_{i}\right\rangle -\left\langle W_{i}\right\rangle \right|\right),
\end{equation}
and the system has a noncontextual description if and only if\emph{
}\textup{\emph{
\begin{equation}
\so\left(\left\langle V_{i}W_{i\oplus_{3}1}\right\rangle :i\in\left\{ 1,2,3\right\} \right)\leq1+\sum_{i=1}^{3}\left|\left\langle V_{i}\right\rangle -\left\langle W_{i}\right\rangle \right|.
\end{equation}
}}\end{thm}
\begin{proof}
The computer-assisted proof is based on Lemma \ref{lem:1-2-1-1} below,
and its details, omitted here, can be found in Ref. \cite{DK_Bell_LG2},
Appendix and Theorems 4 and 6.\end{proof}
\begin{lem}
\label{lem:1-2-1-1}The necessary and sufficient condition for the
connections $\left\{ \left(V_{i},W_{i}\right):i\in\left\{ 1,2,3\right\} \right\} $
to be compatible with the observed pairs $\left\{ \left(V_{i},W_{i\oplus_{3}1}\right):i\in\left\{ 1,2,3\right\} \right\} $
is 
\begin{equation}
\so\left(\left\langle V_{i}W_{i\oplus_{3}1}\right\rangle ,\left\langle V_{i}W_{i}\right\rangle :i\in\left\{ 1,2,3\right\} \right)\leq4,
\end{equation}
which can be equivalently written as

\begin{equation}
\begin{array}{l}
\so\left(\left\langle V_{i}W_{i\oplus_{4}1}\right\rangle :i\in\left\{ 1,2,3\right\} \right)+\se\left(\left\langle V_{i}W_{i}\right\rangle :i\in\left\{ 1,2,3\right\} \right)\le4,\\
\\
\se\left(\left\langle V_{i}W_{i\oplus_{4}1}\right\rangle :i\in\left\{ 1,2,3\right\} \right)+\so\left(\left\langle V_{i}W_{i}\right\rangle :i\in\left\{ 1,2,3\right\} \right)\le4.
\end{array}\label{eq:compatibility-1}
\end{equation}

\end{lem}

\subsection{Special cases}

If the SZLG-system is a CC-system, $\Delta_{0}\left(\mathfrak{C}\right)=0$,
the criterion for noncontextuality acquires the form
\begin{equation}
\so\left(\left\langle V_{i}W_{i\oplus_{3}1}\right\rangle :i\in\left\{ 1,2,3\right\} \right)\leq1.
\end{equation}
This can be written in the more familiar (Suppes-Zanotti's) form \cite{SuppesZanotti1981}
as
\begin{equation}
-1\le\left\langle V_{1}W_{2}\right\rangle +\left\langle V_{2}W_{3}\right\rangle +\left\langle V_{3}W_{1}\right\rangle \leq1+2\max\left(\left\langle V_{1}W_{2}\right\rangle ,\left\langle V_{2}W_{3}\right\rangle ,\left\langle V_{3}W_{1}\right\rangle \right).
\end{equation}

\begin{rem}
In the temporal version of SZLG-systems (the Leggett-Garg paradigm
proper), $V_{1}$ and $W_{1}$ are the results of the first measurement
in both $\left(V_{1},W_{2}\right)$ and $\left(V_{3},W_{1}\right)$.
They therefore cannot be influenced by later measurements (no signaling
back in time). Consequently, if there are no contextual measurement
biases, $V_{1}\sim W_{1}$, and 
\begin{equation}
\Delta_{0}\left(\mathfrak{C}\right)=\frac{1}{2}\sum_{i=2}^{3}\left|\left\langle V_{i}\right\rangle -\left\langle W_{i}\right\rangle \right|.
\end{equation}
Including $\left|\left\langle V_{1}\right\rangle -\left\langle W_{1}\right\rangle \right|$,
however, does not hurt, and it allows one to accommodate cases with
non-temporal measurements and biased measurements (when knowing whether
variable 1 will be paired with $2$ or with 3 changes the way one
measures 1).
\end{rem}

\section{Comparing the Systems}

The main theorems regarding our three systems, Theorems \ref{thm:main B},
\ref{thm:main B-1}, and \ref{thm:main B-1-1}, strongly suggest the
following generalization, which we formulate as a conjecture.
\begin{rem}
As of February 2015 we have a proof of this conjecture. A proof of
the supporting Lemma \ref{lem:supporting} is given in Ref. \cite{KDL2014}.\end{rem}
\begin{conjecture}
\label{conj:Gen}Let $\left(\mathfrak{S},\mathfrak{C}\right)$ be
a system with
\begin{equation}
\mathfrak{S}=\left\{ \left(V_{i},W_{\pi\left(i\right)}\right):i\in\left\{ 1,\ldots,n\right\} \right\} ,
\end{equation}
where $\pi\left(\left\{ 1,\ldots,n\right\} \right)$ is a circular
(having a single cycle) permutation of $\left\{ 1,\ldots,n\right\} $,
and
\begin{equation}
\mathfrak{C}=\left\{ \left(V_{i},W_{i}\right):i\in\left\{ 1,\ldots,n\right\} \right\} .
\end{equation}
Then
\begin{equation}
\Delta_{0}\left(\mathfrak{C}\right)=\frac{1}{2}\sum_{i=1}^{n}\left|\left\langle V_{i}\right\rangle -\left\langle W_{i}\right\rangle \right|,
\end{equation}
and 
\begin{equation}
\Delta_{\min}\left(\mathfrak{S},\mathfrak{C}\right)=\frac{1}{2}\max\left(2\Delta_{0}\left(\mathfrak{C}\right),\;\so\left(\left\langle V_{i}W_{\pi\left(i\right)}\right\rangle :i\in\left\{ 1,\ldots,n\right\} \right)-n+2\right).
\end{equation}
Consequently, the degree of contextuality in this system is
\begin{equation}
\cntx\left(\mathfrak{S},\mathfrak{C}\right)=\frac{1}{2}\max\left(0,\so\left(\left\langle V_{i}W_{\pi\left(i\right)}\right\rangle :i\in\left\{ 1,\ldots,n\right\} \right)-n+2-\sum_{i=1}^{n}\left|\left\langle V_{i}\right\rangle -\left\langle W_{i}\right\rangle \right|\right),\label{eq:gen measure}
\end{equation}
and the system has a noncontextual description if and only if\emph{
}\textup{\emph{
\begin{equation}
\so\left(\left\langle V_{i}W_{\pi\left(i\right)}\right\rangle :i\in\left\{ 1,\ldots,n\right\} \right)\leq n-2+\sum_{i=1}^{n}\left|\left\langle V_{i}\right\rangle -\left\langle W_{i}\right\rangle \right|.\label{eq:gen criterion}
\end{equation}
}}
\end{conjecture}
The corresponding generalization of the supporting lemmas is
\begin{lem}
\label{lem:supporting}The necessary and sufficient condition for
the connections $\left\{ \left(V_{i},W_{i}\right):i\in\left\{ 1,\ldots,n\right\} \right\} $
to be compatible with the observed pairs $\left\{ \left\langle V_{i},W_{\pi\left(i\right)}\right\rangle :i\in\left\{ 1,\ldots,n\right\} \right\} $
is 
\begin{equation}
\so\left(\left\langle V_{i}W_{\pi\left(i\right)}\right\rangle ,\left\langle V_{i}W_{i}\right\rangle :i\in\left\{ 1,\ldots,n\right\} \right)\leq2n-2,
\end{equation}
which can be equivalently written as
\begin{equation}
\begin{array}{l}
\so\left(\left\langle V_{i}W_{\pi\left(i\right)}\right\rangle :i\in\left\{ 1,\ldots,n\right\} \right)+\se\left(\left\langle V_{i}W_{i}\right\rangle :i\in\left\{ 1,\ldots,n\right\} \right)\le2n-2,\\
\\
\se\left(\left\langle V_{i}W_{\pi\left(i\right)}\right\rangle :i\in\left\{ 1,\ldots,n\right\} \right)+\so\left(\left\langle V_{i}W_{i}\right\rangle :i\in\left\{ 1,\ldots,n\right\} \right)\le2n-2.
\end{array}\label{eq:compatibility-1-1}
\end{equation}

\end{lem}
It is easy to see that the criterion and measure for the SZLG-system
(Theorem \ref{thm:main B-1}) is a special case of those for the EPRB-system
(Theorem \ref{thm:main B-1-1}) which in turn is a special case of
those for the KCBS system. Specifically, by putting $\left\langle V_{5}W_{1}\right\rangle =1$
in the KCBS-system, and assuming in addition that $W_{5}\sim V_{5}$,
so that $\left\langle V_{5}W_{5}\right\rangle =1$ in the maximal
coupling, we can replace $W_{5}$ in $\left\langle V_{4}W_{5}\right\rangle $
with $W_{1}$ and obtain the EPRB-system. By putting then $\left\langle V_{4}W_{1}\right\rangle =1$
in the EPRB-system, and $W_{4}\sim V_{4}$, so that $\left\langle V_{4}W_{4}\right\rangle =1$
in the maximal coupling, can replace $W_{4}$ in $\left\langle V_{3}W_{4}\right\rangle $
with $W_{1}$ and obtain the SZLG-system.

It is easy to see how this pattern generalizes. First of all, any
circular permutation $\pi$ can be replaced, by appropriate renaming,
with circular addition of 1 on $\left\{ 1,\ldots,n\right\} $, making
the observed pairs $\left(V_{1},W_{2}\right),\ldots,\left(V_{n-1},W_{n}\right),\left(V_{n},W_{1}\right)$.
Let the SZLG, EPRB, and KCBS systems be designated as systems of order
3, 4, and 5, respectively, and the system in Conjecture \ref{conj:Gen}
as an $\left(n\right)$-system. By putting $\left\langle V_{n}W_{1}\right\rangle =1$
in the $(n)$-system, and assuming that $W_{n}\sim V_{n}$ , so that
$\left\langle V_{n}W_{n}\right\rangle =1$ in the maximal coupling,
we replace $W_{n}$ in $\left\langle V_{n-1}W_{n}\right\rangle $
with $W_{1}$ and obtain an $\left(n-1\right)$-system.
\begin{rem}
In Ref. \cite{KDL2014} we have proved the following theorem: the
system $\left(\mathfrak{S},\mathfrak{C}\right)$ in Conjecture \eqref{conj:Gen}
has a noncontextual description if and only if
\begin{equation}
\so\left(\left\langle V_{i}W_{\pi\left(i\right)}\right\rangle ,1-\left|\left\langle V_{i}\right\rangle -\left\langle W_{i}\right\rangle \right|:i=1,\ldots,n\right)\le2n-2.\label{eq:master}
\end{equation}
It is easy to show that the conjectured criterion \eqref{eq:gen criterion}
follows from \eqref{eq:master}, i.e., its violation is a sufficient
condition for contextuality. The conjecture is that it is also true
that \eqref{eq:master} follows from \eqref{eq:gen criterion}. See
Remark 40.
\end{rem}

\section{Conclusion}

We have presented a theory of (non)contextuality in purely probabilistic
terms abstracted away from physical meaning. The computational aspects
of the theory are confined to finite systems with binary components,
but it is easily generalizable to deal with components attaining arbitrary
finite numbers of values. As the components of the bunches get more
complex and/or connections get longer than pairs, the generalizations
become less unique.

The basis for the theory is the principle of Contextuality-by-Default,
which has philosophical and mathematical consequences. Mathematically,
it leads to revamping (while remaining within its confines) of the
Kolmogorovian probability theory, with more prominent than usual emphasis
on stochastic unrelatedness. A reformulated theory may even avoid
the notion of a sample space altogether \cite{8NHMT}. 

Philosophically, the principle of Contextuality-by-Default may elucidate
the difference between ``ontic'' and ``epistemic'' aspects of
contextuality (perhaps even of probability theory generally).

The theory also offers pragmatic advantages: it allows for non-CC-systems
(whether due to signaling or due to context-dependent measurement
biases) and for experimental/computational errors in the analysis
(including statistical analysis) of experimental data. 

The theory has predecessors in the literature. The idea of labeling
differently random variables in different contexts and considering
the probability with which they can be equal to each other if coupled
has been prominently used in Refs. \cite{2Larsson,Svozil}.

\subsubsection*{Acknowledgments}

This work is supported by NSF grant SES-1155956 and AFOSR grant FA9550-14-1-0318.
We have benefited from collaboration with J. Acacio de Barros, and
Gary Oas, as well as from discussions with Samson Abramsky, Guido
Bacciagaluppi, and Andrei Khrennikov. An abridged version of this
paper was presented at the Purdue Winer Memorial Lectures in November
2014.

\end{document}